\RequirePackage{fix-cm}

\documentclass[12pt, twoside]{article}
\usepackage{amsmath,amsthm,amssymb}
\usepackage{times}
\usepackage{enumerate}

\theoremstyle{definition}
\newtheorem{thm}{Theorem}[section]
\newtheorem{cor}[thm]{Corollary}
\newtheorem{lem}[thm]{Lemma}

\newtheorem{prop}[thm]{Proposition}
\newtheorem{rmk}[thm]{Remark}
\newtheorem*{prob}{Problem}
\newtheorem*{thm*}{Theorem}

\numberwithin{equation}{section}

\newcommand{\subjclass}[1]{\bigskip\noindent\emph{2010 Mathematics Subject Classification:}\enspace#1}
\newcommand{\keywords}[1]{\noindent\emph{Keywords:}\enspace#1}

\usepackage{graphicx}
\usepackage{color}
\usepackage{mathptmx}
\usepackage{natbib}
\usepackage{enumerate}
\usepackage{amsmath,amssymb}
\usepackage{verbatim}
\usepackage[misc]{ifsym}

\definecolor{alert}{rgb}{0.8,0,0.3}
\newcommand{\alert}[1]{%
	\marginpar{%
		\ifodd\value{page} \raggedright \else \raggedleft \fi
		\footnotesize{\textcolor{alert}{#1}}
	}
}

\newcommand{\R}{\mathbb{R}}
\newcommand{\N}{\mathbb{N}}
\newcommand{\BH}{{\mathbf H}}
\newcommand{\SH}{\mathcal{H}}
\renewcommand{\S}{\mathbb{S}}
\newcommand{\IP}[2]{\left< #1 , #2 \right>}
\newcommand{\vn}[1]{\lVert#1\rVert}
\newcommand{\A}{\text{Area }\Sigma}
\newcommand{\V}{\text{Vol }\Sigma}

\usepackage[hidelinks]{hyperref}
\newcommand*{\email}[1]{\href{mailto:#1}{\nolinkurl{#1}} } 

\begin{document}


\baselineskip=17pt


\title{Rigidity and stability of spheres in the Helfrich model}

\author{Yann Bernard\\
              Department Mathematik, ETH Z\"urich, Switzerland \\
              \email{yann.bernard@math.ethz.ch} \\
        Glen Wheeler (\Letter)\\
              School of Mathematics and Applied Statistics, University of Wollongong, Australia \\
              \email{glenw@uow.edu.au}\\
        Valentina-Mira Wheeler\\
              School of Mathematics and Applied Statistics, University of Wollongong, Australia \\
              \email{vwheeler@uow.edu.au}
      }
\date{}


\maketitle

\begin{abstract}
The Helfrich functional, denoted by $\SH^{c_0}$, is a mathematical expression
proposed by \citet{H73} for the natural free energy carried by an elastic
phospholipid bilayer.
Helfrich theorises that idealised elastic phospholipid bilayers minimise
$\SH^{c_0}$ among all possible configurations.  The functional integrates a
spontaneous curvature parameter $c_0$ together with the mean curvature of the
bilayer and constraints on area and volume, either through an inclusion of
osmotic pressure difference and tensile stress or otherwise.
Using the mathematical concept of embedded orientable surface to represent the
configuration of the bilayer, one might expect to be able to adapt methods from
differential geometry and the calculus of variations to perform a fine analysis
of bilayer configurations in terms of the parameters that it depends upon.
In this article we focus upon the case of spherical red blood cells with a view
to better understanding spherocytes and spherocytosis.
We provide a complete classification of spherical solutions in terms of the
parameters in the Helfrich model.
We additionally present some further analysis on the rigidity and stability of
spherocytes.

\keywords{Spherocytosis \and Biomembranes \and Helfrich model \and Differential geometry}
\subjclass{74K15 \and 51P05 \and 00A71}
\end{abstract}

\section{Introduction}
\label{Introduction}

Motivated by Hooke's law, \citet{H73} proposed
\[
f_c = \frac{k_c(H-c_0)^2}{2} + \overline{k}K
\]
as the energy per unit area of a lipid bilayer or membrane.
The constants $k_c$ and $\overline{k}$ are the bending moduli.
It is argued in \citet{MH90,DKS90} that $k_c$ is small and positive.
We shall see shortly that the exact value of $\overline{k}$ is not important for our investigations
here (in \citet{S97} it is even set to zero).
The spontaneous curvature $c_0$ on the other hand is a critical component of the model.
It was found, based on experimental data of \citet{EF72}, to approximately satisfy $c_0 = -0.74\,\mu
m^{-1}$ under the assumptions that the membrane is a typical human erythrocyte and normal
physiological conditions are in place (see \citet{DH76,DH76p2}).
For more details on values for these and other parameters we refer to the Remark after
Theorem 1.

Supposing the membrane is represented by a smooth isometric embedding $f:\Sigma\rightarrow\R^3$ of a
two-dimensional closed differentiable manifold $\Sigma$, this gives rise to the energy functional
\begin{align*}
\hat{\SH}^{c_0}(f)
 &= \frac{k_c}{2} \int_\Sigma (H-c_0)^2d\mu + 2\overline{k}\pi\chi(\Sigma)
\\
 &=  \frac{k_c}{2} \int_\Sigma H^2d\mu 
   - k_cc_0 \int_\Sigma H\,d\mu
   + \frac{k_cc_0^2}{2} \text{Area }f
   + 2\overline{k}\pi\chi(\Sigma)\,,
\end{align*}
where $\chi(\Sigma)$ is the Euler characteristic of $\Sigma$ and we have used the Gauss-Bonnet
theorem.
The Euler characteristic $\chi(\Sigma)$ is a topological invariant, satisfying for example
$\chi(\Sigma) = 2 - 2g$, where $g$ is the genus of $\Sigma$.
The genus counts the number of holes in the surface.
In cases where the bilayer is topologically spherical, we have genus$(\Sigma) = 0$.
This includes for example stomatocytes, discocytes, spherocytes, echinocytes, and so on (see \cite{bloodcells}).
As such cells form our primary interest in this paper, we shall work from now on in the topological
class of $g=0$ and $\chi(\Sigma) = 2$.
The notation Area $f$ denotes the area of $(\Sigma, f^*\IP{\cdot}{\cdot})$ as a Riemannian manifold,
where $f^*\IP{\cdot}{\cdot}$ is the pullback via the embedding $f$ of the standard metric
$\IP{\cdot}{\cdot}$ on $\R^3$, the dot product.
That is,
\[
\text{Area }f = \int_\Sigma d\mu = \int_\Sigma \sqrt{\text{det
}(\IP{\partial_if}{\partial_jf}})\,dx\,.
\]
Our motivation for the study of the Helfrich model is in connection with
spherocytosis, a disorder of the membrane of human red blood cells that causes
them to be spherical (spherocytes) as opposed to the standard biconcave disk
shape (discocytes).
Spherocytes break down faster than discocytes, and as they have a lower surface
area than discocytes, (in fact by the isoperimetric problem, spherocytes are in
this sense the worst configuration possible) patients with spherocytosis suffer
from severe anemia \citep{chasis1988decreased,medPGM08,svetina1989membrane}.
Additionally, the spleen sometimes mistakes otherwise healthy spherocytes for
damaged non-functional cells, and destroys them.
This leads to haemolytic anemia, and can be fatal \citep{medPGM08}.
Up to now, the only known treatment is a (often partial) splenectomy
\citep{medAZCRBCCNC09,medHP96,medPGM08,medRWEG07},
which comes with a lifetime of medication, and other complications.
Spherocytosis is the most common form of inheritable anemia in people of
northern European ancestry \citep{medPGM08}.

The cytoskeleton of a human red blood cell is inhomogeneous and sheet-like,
with a lipid bilayer and supporting network of proteins.
Although there remain many open questions regarding the dynamical forces at
play in the cytoskeleton of a human red blood cell (see \cite{steck}), it
is in a sense self-organising so as to minimise certain costs (see \cite{constantc01}).
The basic idea of the Helfrich model is that this cost can be measured in an
idealised setting by the Helfrich functional.
Our goal is to better understand the appearance of spherocytes \emph{in the
model}.
In particular, we study spherical solutions of the Euler-Lagrange equation for
critical points of $\SH^{c_0}$, giving first a complete classifcation of
parameter ranges that allow spherocytes (Theorem \ref{TMspheresintro}) and
second some first steps into rigidity and stability analysis of spherical
solutions (Theorema 2, 4, 8, 9 and Corollaries 5 and 6).
Since the parameters $c_0$, $\lambda$, and $p$ are in principle measurable, we
may in the long-term be able to influence them, and in so doing discourage the
formation of spherocytes.

We now seek to study embeddings $f:\Sigma\rightarrow\R^3$ that \emph{minimise}
the Helfrich functional.
These represent, in a model sense, the biomembranes that we wish to investigate.

In the language of the calculus of variations, the problem is then as follows.
\begin{prob}[P1]
\label{CVPlagmultmethod}
Suppose $\Sigma$ is a closed differentiable 2-manifold with genus zero.
Let $c_0$, $S_0$ and $V_0$ be fixed positive constants.
Minimise $\hat{\SH}^{c_0}(f)$ in the class of smooth embeddings $f:\Sigma\rightarrow\R^3$ subject to
the constraints
\begin{equation}
\label{P1constraint}
\tag{P1.1}
\text{Area }f = S_0\quad\text{and}\quad \text{Vol }f = V_0\,.
\end{equation}
That is, find an embedding $f_0:\Sigma\rightarrow\R^3$ such that Area $f = S_0$, Vol $f = V_0$, and
\begin{equation}
\label{EQmin1}
\hat{\SH^{c_0}}(f_0) \le 
\hat{\SH^{c_0}}(f)
\end{equation}
for any other smooth embedding $f$ of $\Sigma$.
\end{prob}

\begin{rmk}
A candidate embedding $f_0$ which achieves the global energy minimum is called a \emph{solution}. It
is not unique. The constraints and the functional $\hat{\SH}$ are invariant under reparametrisation
as well as rigid motions in $\R^3$.
\end{rmk}

The variational problem (P1) is the classical formulation suggested in \citet{H73,DH76,DH76p2}.
A solution $\hat{f}:\Sigma\rightarrow\R^3$ will satisfy the aforementioned
Euler-Lagrange equation
\begin{equation}
\label{EQeulerlagrange2}
k_c(\Delta H + H|A^o|^2) + 2k_cc_0K - \Big(\frac{k_cc_0^2}{2}+s_0\Big)H - v_0 = 0
\,.
\end{equation}
where $s_0,v_0\in\R$ are Lagrange multipliers (see \cite{capovilla}).
In the above we have used $\Delta = g^{ij}\nabla_i\nabla_j$ to denote the Laplace-Beltrami operator and
$A^o = A - \frac12gH$, where $g$ is the metric induced by $\hat{f}$, to denote
the tracefree part of the second fundamental form $A$.
For more details on our notation we refer the reader to Section 2.

We emphasise that $s_0$ and $v_0$ in \eqref{EQeulerlagrange2} are abstract
mathematical constants; they have no physical meaning.
Their role is to ensure that the restrictions \eqref{P1constraint} are satisfied by $\hat{f}$.
They do not represent any physical force in the original formulation.

It is possible to derive an expression similar to \eqref{EQeulerlagrange2} for
the shape of a biomembrane where constants with a possible physical meaning appear
in a manner identical to $s_0$ and $v_0$.
This can be achieved via the inclusion of the osmotic pressure difference $p$
and tensile stress $\lambda$ in the expression for the free energy of a closed
bilayer.
This slightly different approach has by now become quite common -- see for example
\citet{B13,T06} and \citet{Vnotes}.
Indeed, \citet{DH76} argued that the Lagrange multipliers $s_0$ and $v_0$ above
essentially play these roles.
This leads to the alternative functional
\begin{align*}
\SH^{c_0}(f)
 &=  \frac{k_c}{2} \int_\Sigma (H-c_0)^2d\mu 
   + \lambda\text{Area }\Sigma
   + p\text{Vol }\Sigma
   + 2\overline{k}\pi\chi(\Sigma)
\\
 &=  \frac{k_c}{2} \int_\Sigma H^2d\mu 
   - k_cc_0 \int_\Sigma H\,d\mu
   + \Big(\frac{k_cc_0^2}{2}+\lambda\Big) \text{Area }\Sigma
   + p\text{Vol }\Sigma
   + 2\overline{k}\pi\chi(\Sigma)\,.
\end{align*}
It is the functional $\SH^{c_0}$ above that we study in this paper.
For clarity, we restate the minimisation problem associated to this functional below.
\begin{prob}[P2]
\label{CVPaugmentedfunctional}
Suppose $\Sigma$ is a closed differentiable 2-manifold with genus zero.
Let $c_0$, $p$ and $\lambda$ be fixed constants.
Minimise ${\SH}^{c_0}(f)$ in the class of smooth embeddings $f:\Sigma\rightarrow\R^3$.
That is, find an embedding $f_0:\Sigma\rightarrow\R^3$ such that
\begin{equation}
\label{EQmin2}
{\SH^{c_0}}(f_0) \le 
{\SH^{c_0}}(f)
\end{equation}
for any other smooth embedding $f$ of $\Sigma$.
\end{prob}
A solution to (P2) will satisfy the Euler-Lagrange equation \eqref{EQeulerlagrange2} with tensile
stress $\lambda$ substituted for $s_0$ and the osmotic pressure difference $p$ substituted for
$v_0$:
\begin{equation}
\label{EQeulerlagrange}
\BH^{c_0}(f) := k_c(\Delta H + H|A^o|^2) + 2k_cc_0K - \Big(\frac{k_cc_0^2}{2}+\lambda\Big)H - p = 0
\,.
\end{equation}

\begin{rmk}
In problems (P1) and (P2), the spontaneous curvature $c_0$ is stated to be a fixed constant.
From a physical perspective, this is not the case, since it is known that the spontaneous curvature
$c_0$ has the units of one of the principal cuvatures of the membrane $f$.
In particular, given a membrane $f:\Sigma\rightarrow\R^3$ with spontaneous curvature $c_0$, a
dilated membrane $\rho f$ has spontaneous curvature $\frac{c_0}{\rho}$.
We may further assume that the spontenous curvature is invariant under rigid motions in $\R^3$, i.e.
translations and rotations.
Since it is agreed in the literature that due to the homogeneity of the membrane the spontaneous
curvature does not vary based on position, we could assume that
\[
c_0 = \int_\Sigma F[f,\vec{v}]\,d\mu
\]
for an operator $F$ and a vector of parametric functions $\vec{v}(p) =
(\vec{\hat{v}}\circ f)(p)$ where $\vec{\hat{v}} = (v_1(x),\ldots,v_M(x))$,
$x\in\R^3$, $m\in\N$.
The role of the vector $\vec{v}$ would be to incorporate ambient information
into the determination of the spontaneous curvature.
A similar procedure was enacted in \cite{HelfrichWheeler} for a model of
strings in space influenced by ambient forces.
The known behaviour of $c_0$ under dilation and rigid motions would translate
to the operator $F$ being invariant under rigid motions and homogeneous of
degree $-3$.

Unfortunately this appears to be the most that is known.
It is not clear from the literature how exactly the spontaneous curvature
depends upon the embedding $f$.
This is important, since the existence of one or more solutions to problems
(P1) and (P2) and the qualitative properties such solutions possess would
depend critically on the structure of $F$.
Discovering new properties and further information on the nature of $F$ is an
important open problem in the field.
\end{rmk}

\begin{rmk}
It is physically reasonable to require that the energy of a biomembrane not
depend on scale -- this amounts to the requirement that $\SH^{c_0}(f) =
\SH^{c_0}(\rho\,f)$.
Taking into account the scaling of the measure $d\mu$, the mean curvature $H$,
and the volume, the units of $c_0$, $\lambda$ and $p$ should be $H^2\,d\mu$,
$(d\mu)^-1$ and $(d \text{Vol})^{-1}$ respectively.
In terms of $\rho$, this is $\rho^{-1}$, $\rho^{-2}$ and $\rho^{-3}$.
The parameters $k_c$ and $\overline{k}$ should be scale invariant, or, the
behaviour of $\lambda$ and $p$ under scaling should incorporate information on
how $k_c$ and $\overline{k}$ scale.
A precise formulation of the Helfrich model taking into account such scale
invariance does not yet appear to be available, although it is implicit in
the Lagrange multipliers of (P1).
\end{rmk}

\begin{rmk}
From an analysis perspective, questions on existence and regularity of
solutions to \eqref{EQeulerlagrange} must be investigated.
For the Willmore functional, where $c_0 = \lambda = p = 0$, this is a venerable
topic.
\cite{B84} classified all closed solutions through a duality method.
A landmark contribution in existence was made by \cite{S93}.
Remarkable progress on regularity issues was made by \cite{R}, who decoupled
\eqref{EQeulerlagrange} into two second order systems and studied weak
solutions.
Both \cite{BRsing} and \cite{KS04} made important contributions to the
understanding of point singularities.
The Willmore conjecture, proposed by \cite{Willmore65}, was recently resolved
by \cite{NMconjecture}.
Work on the Willmore functional continues to be a very active area, with recent
progress made on quantisation \citep{BRquant}, the gradient flow
\citep{KS01,KS02}, and boundary value problems \citep{AK14,D12,DDW,DG09}.
There are many other works besides those mentioned here -- the literature on
analysis of the Willmore functional is vast.
Work on intermediate functionals, both from a numerical and theoretical
standpoint, has also been active -- the workshop \citep{GarckeOber} and articles
\citep{GarckeLocal,GarckePara} are an excellent resource on this.
For the full Helfrich functional, many of these issues remain open\footnote{A
partial answer to the existence and regularity question can be found in
\citet{C13}.} and form important questions that future research should address.
\end{rmk}

Analysis of solutions to \eqref{EQeulerlagrange2} is quite involved.
Here our reduced focus allows us to pin down the influence of the spontaneous
curvature $c_0$ on solutions to \eqref{EQeulerlagrange} for any value of $c_0$.
The classification theorem is as follows.
It is proved in Section 3.
Note that we use below and throughout the paper $S_r:\S^2\rightarrow\R^3$ to
denote the standard embedding of a sphere with (typically) unspecified centre.

\begin{thm}
\label{TMspheresintro}
Suppose $f:\Sigma\rightarrow\R^3$ is a closed, smooth, embedded orientable surface in the
same topological class as a sphere.
If $f(\Sigma) = S_r(\Sigma)$, and $f$ is critical for the Helfrich functional $\SH^{c_0}$, then one of the following must hold:
\begin{enumerate}
	\item[(a)] For $c_0 = 0$:
		\begin{enumerate}
	\item[(i)]   $c_0 = \lambda = p = 0$, in which case $f$ may be a sphere of any radius;
	\item[(ii)]  $p\lambda < 0$, in which case $f$ must be the unique critical sphere with radius $r = -\frac{2\lambda}{p}$;
		\end{enumerate}
	\item[(b)] If $c_0 \ne 0$, set $u := \frac{\lambda}{k_cc_0} + \frac{c_0}{2}$ and $v := \frac{2p}{k_cc_0}$. Then:
		\begin{enumerate}
	\item[(iii)] $v = -u^2$, $u>0$, in which case $f$ must be the unique critical sphere with radius $r = \frac{u}{2}$;
	\item[(iv)]  $u \ge 0$ and $v\ge0$, in which case $f$ must be the unique critical sphere with radius $r = \frac12(u + \sqrt{u^2 + v})$;
	\item[(v)]   $u \ge 0$ and $v\in(-u^2,0)$, in which $f$ may be either of the two critical spheres with radii $r_\pm = \frac12(u \pm \sqrt{u^2 + v})$;
	\item[(vi)]  $u < 0$ and $v>0$, in which case $f$ must be the unique critical sphere with radius $r = \frac12(u + \sqrt{u^2 + v})$.
		\end{enumerate}
\end{enumerate}
\end{thm}

\begin{rmk}[Experimental determination of parameters]
While it is possible to obtain experimental values for $c_0$ for vesicles (as
in the work of \cite{EF72} mentioned earlier, see also \cite{values1}), there
currently exist no direct measurements of the spontaneous curvature of red
blood cells. One may reasonably guess that $c_0\ne0$ as the distribution of the
phospholipid types between the two leaflets of the bilayer is asymmetric (see
\cite{bloodcells}).

Although the spontaneous curvature should realistically depend on position $c_0
= c_0(x)$ where $x\in\Sigma$ (one possible choice of function is to use the
mean curvature of a `resting shape', see \cite{constantc01,nonconstantc01} for example), we
take the view here that $c_0$ is constant.
\cite{constantc01} suggests that an appropriate choice for the spontaneous
curvature of normal red blood cells is $c_0 = -0.62\, \mu m^{-1}$, very similar
to that reported by \citet{DH76,DH76p2}, where a range of values that confirm
observed experimental data is given, from $-0.56\, \mu m^{-1}$ to $-1.94\, \mu
m^{-1}$. Later the value of $-0.74\,\mu m^{-1}$ was settled upon.

It may be interesting to note that if one considers a sphere as the resting
shape for red blood cells then these two approaches align, but this only works
for positive spontaneous curvature.
We have taken $c_0$ to be constant, allowing positive and negative values, for
simplicity.  Allowing $c_0$ to depend on position is an important topic for
future work.

The bending modulus $k_c$ can be measured experimentally, although reported
values for healthy human red blood cells vary in the range 0.2 -- 9.0 $\times
10^{-19}$ J (see \cite{bending1} and \cite{bending2} for example). It is
reported in \cite{guckenberger17values} that most simulations are being
performed for $k_c$ between 2 and 4 $\times 10^{-19}$ J.

In \cite{constraints} a similar model to ours here is studied. First, we must
repeat that both $\lambda$ and $p$ should be considered as functions, rather
than constants, and here we have treated them as such only for simplicity.
The osmotic pressure difference $p$ for a human red blood cell has been argued
by \citet{DH76,DH76p2} to be small when approaching spherocytes, but of indeterminate sign.
When $p$ is relatively large, discrete rotational symmetries in addition to
reflection tend to appear in equilibria, whereas for negative $p$ all discrete
symmetries, including reflection symmetry, appear to be lost (see \cite{DH76}).
For the tensile stress $\lambda$ \cite{DH76} show that $\lambda$ should be of
the order $-p/k_c$, and since $k_c$ is positive, this means that $\lambda$
should have sign opposite to that of $p$.
In general $p$ and $\lambda$ are expected to depend on the equilibrium
configuration; it may not be reasonable to prescribe them a-priori.

Regarding the parameters $p$ and $\lambda$, \cite{constraints} state {\it
``Solution of the shape equation is complicated by the fact that the values of
the parameters $\lambda$ and $p$ are not known at the outset.''} They continue
to note that in practice, these parameters are tweaked after simulation until
the desired constraints on area and volume are satisfied.
Their dependence on area and volume is not clear, although, in some cases it is
possible to identify when they penalise or reward area and volume growth or
decay, see the discussion in \cite{guven} on surface tension for example.

There exists a heuristic `rule' that Lagrange multipliers are equal to the
derivative of the functional under constraint with respect to the quantity
being constrained.
This would imply that $\lambda$ is the derivative of $\int_\Sigma
(H-c_0)^2\,d\mu$, evaluated at extrema, with respect to $\A$.
Similarly, $p$ would be the derivative of $\int_\Sigma (H-c_0)^2\,d\mu$,
evaluated at extrema, with respect to $\V$.
However such a derivative is not well-defined, as almost nothing is known about
the manifold of equilibria of the functional $\int_\Sigma (H-c_0)^2\,d\mu$;
furthermore it seems unlikely that such a derivative could ever be
well-defined.

This reveals one core philosophical difference between (P1) and (P2): For (P1),
the Lagrange multipliers may never enjoy a physical interpretation, whereas for
(P2) the parameters $\lambda$ and $p$ remain open to interpretation, leaving
hope that they may be physically relevant and measurable.

There is no consensus among the literature even on the values of the heavily
studied bending modulus, or spontaneous curvature.
One may view our results here, that focus on spherical critical points for the
Helfrich functional $\SH^{c_0}$, as further informing the discussion on
realistic values for $\lambda$ and $p$.

Bending forces in Helfrich's spontaneous curvature model have been recently
surveyed by \cite{guckenberger17values}, where one may find many further
details and simulations in the above directions.
\end{rmk}

The question of when a critical point for the functional $\SH^{c_0}$ is a sphere is much more delicate.
We present the following result, which is a straightforward consequence of the uniqueness of embedded
CMC surfaces (the Hopf theorem) and the isoperimetric inequality.

\begin{thm}
\label{mainthmx}
Spheres are the unique global minimisers of the energy $\SH^{c_0}$ among closed
surfaces in the same topological class as a sphere with volume fixed at 
$\V = \frac{32\pi}{3c_0^3}$,
if $c_0, \lambda, p$ satisfy:
\begin{equation}
\label{mainthmxcondn}
3\lambda  + \frac{p}{c_0} = 0\,.
\end{equation}
\end{thm}

This result requires the very restrictive assumptions that the volume be fixed
at a level that includes the sphere with radius $\frac1{c_0}$, and the
parameters of the functional satisfy \eqref{mainthmxcondn}.
Relaxing these conditions can be achieved by introducing a closeness assumption.

Such a closeness assumption has been used by \cite{MW13} in the case of zero spontaneous
curvature.
In \citet{MW13}, the functional
\[
\tilde{\SH}^{c_0}_{\lambda_1,\lambda_2}(f)
 = \frac14\int_\Sigma (H-c_0)^2d\mu + \lambda_1 \text{Area }\Sigma + \lambda_2 \text{Vol }\Sigma
\]
was studied.
This is far from $\hat{\SH}^{c_0}$ but is on the other hand quite close to the functional $\SH^{c_0}$
featuring in problem (P2), with many properties common to both $\tilde{\SH}^{c_0}$ and $\SH^{c_0}$.
The difference between the two is given by:
\[
\SH^{c_0}(f) - 2k_c\tilde{\SH}^{c_0}_{\frac{\lambda}{2k_c},\frac{p}{2k_c}}(f)
 = 
     2\overline{k}\pi\chi(\Sigma)
\,,
\]
which is constant.
For the case where the bilayer is topologically spherical we have $\chi(\Sigma) = 2$ and
\begin{equation}
\label{EQdifference}
\SH^{c_0}(f) - \big(2k_c\tilde{\SH}^{c_0}_{\frac{\lambda}{2k_c},\frac{p}{2k_c}}(f)
                    + 4\overline{k}\pi\big)
 = 0
\,.
\end{equation}
That is, the functionals $\SH^{c_0}$ and $\tilde{\SH}^{c_0}$, up to taking special choices of the parameters
$\lambda_1$ and $\lambda_2$, differ by a constant.
The variational properties of these functionals are therefore equivalent; only the numerical energy
of shapes is altered.
The following theorem is known:

\begin{thm}[Theorem 1 in \cite{MW13}]
Let $f:\Sigma\rightarrow\R^3$ be a closed, smooth, embedded orientable surface
in the same topological class as a sphere.
Suppose additionally that
\[
	\int_\Sigma |A^o|^2\,d\mu < \varepsilon_1\quad\text{and}\quad \lambda \ge 0\,,
\]
where $\varepsilon_1$ is an explicit universal constant.
Then, if $f$ is critical for the Helfrich functional $\SH^0$, it is a standard round sphere.
\end{thm}

In this paper, we are concerned primarily with non-zero spontaneous curvature.
In addition to Theorem \ref{TMspheresintro} and Theorem \ref{mainthmx} above we
give the following, which extends the main theorem of \cite{MW13} for closed
surfaces to the case of non-vanishing spontaneous curvature.

\begin{thm}
\label{mainthmxx}
Let $f:\Sigma\rightarrow\R^3$ be a closed, smooth, embedded orientable surface.
Assume that $f$ is a Helfrich surface, that is,
\begin{equation}
\label{EQcrit}
\BH^{c_0}(f) = 0
\,.
\end{equation}
There exist universal constants $c_1, c_2, c_3, c_4$ such that if
\[
\int_\Sigma |A^o|^2\,d\mu \le \min\bigg\{
                                    \frac1{2c_1c_3},
                                    \frac1{2c_0^2c_2(\A)}
                                  \bigg\}
                          \,,
\]
\[
c_0^2 + 2\frac{\lambda}{k_c} = 2c_0u \ge 0\,,
\]
and
\[
c_0^2(\A) < \frac1{c_4}\,,
\]
then $f(\Sigma) = S_r(\Sigma)$ is a standard round sphere. 
\end{thm}

\begin{rmk}
The quantity $Q = c_0^2(\A)$ identified in the proof 
is dimensionless, as $c_0$ has the units of the mean curvature, and $\A$ has
the units of $d\mu$. Therefore $c_0^2(\A)$ has the units of $H^2d\mu$, and is
scale-invariant.
\end{rmk}

\begin{rmk}
The constants $c_1, \ldots, c_4$ are universal, but the area of $\Sigma$ is not.
One corollary that interprets Theorem \ref{mainthmxx} is:

\begin{cor}
Let $A_0\in(0,\infty)$.
Consider the class
\[
\mathcal{F} = \Big\{f:\Sigma\rightarrow\R^3\,:\,\text{$f$ is a smooth closed embedded orientable surface with $\A < A_0$}
              \Big\}\,.
\]
Suppose 
$
c_0^2 + 2\frac{\lambda}{k_c} = 2c_0u \ge 0\,,
$
and
$
c_0^2A_0 < \frac1{c_4}\,.
$
There exists a constant $\delta > 0$ depending only on $A_0$ such that
any Helfrich surface $f\in\mathcal{F}$ with $\vn{A^o}_2 < \delta$ is a standard round sphere.
\label{CorRigid}
\end{cor}
\end{rmk}

\begin{rmk}
Theorem \ref{mainthmxx} holds for a range of the parameters $c_0$, $\lambda$ and $p$.
This range is \emph{larger} than the range that parameters that spheres are
critical for, see Theorem \ref{TMspheresintro} for a full classification.
For parameters outside this range, we therefore have a \emph{reverse energy gap} phenomenon:

\begin{cor}
Assume the hypotheses of Corollary \ref{CorRigid} and
in addition that the parameters $c_0$, $\lambda$, and $p$ satisfy one of
\begin{enumerate}
	\item[(i)]   $c_0 = 0$ and $p\lambda\ge0$ with one of $p$, $\lambda$ non-zero;
	\item[(ii)]  $c_0 \ge 0$, $u\ge0$ and $v<-u^2$;
	\item[(iii)] $c_0 \le 0$, $u\le0$ and $v\le0$.
\end{enumerate}
There exists an absolute constant $C>0$ depending only on $A_0$ such that
any Helfrich surface satisfies
\[
\int_\Sigma |A^o|^2\,d\mu > C
\,.
\]
\end{cor}
\end{rmk}
The above results indicate \emph{rigidity} of the sphere.
The question of \emph{stability} of the spherocytes is also important, since,
in a patient with spherocytosis, the spherocytes do not regularly become
singular; they are instead stable and nominally functional, despite being
regularly destroyed by the spleen.
This behaviour is not typical for the Helfrich model; indeed, we expect that
generically perturbed spherocytes revert after a perturbation (if they are ever
formed at all), to a standard discocyte shape.
In general, after a perturbation acts upon a biomembrane, there is no guarantee
that the bilayer will return to a global minimum.
It may focus instead on a stable local minimum, or become singular.

In order to illustrate the general setting, here is a result for the model case
where $c_0 = \lambda = p = 0$ and we are dealing with the Willmore functional
and closed Willmore surfaces.
We say that a smooth isometrically embedded surface $f:\Sigma\rightarrow\R^3$
is weakly mean convex if $H(x) \ge 0$ for all $x\in\Sigma$, and mean convex if
$H(x) > 0$ for all $x\in\Sigma$.
If $K \ge 0$ then weak mean convexity becomes weak convexity and mean convexity
becomes convexity.\footnote{Note that in higher dimensions this is typically
expressed by saying that the second fundamental form is positive semi-definite
or positive definite. The given definition is simpler and agrees with this for
surfaces.}

\begin{prop}[Mean convex closed Willmore surfaces]
\label{PRstab1}
Consider a smoothly embedded closed weakly mean convex orientable surface
$f:\Sigma\rightarrow\R^3$.
Suppose that $c_0 = \lambda = p = 0$.
Then if $f$ is critical for the Helfrich functional $\SH^{c_0}$, it must be a sphere.
\end{prop}

Proposition \ref{PRstab1} yields a stability statement in the following sense.
Let $S_r:\Sigma\rightarrow\R^3$ be a standard sphere centred at the origin with radius $r$.
Then it is critical for the Helfrich functional with $c_0 = \lambda = p = 0$.
Consider, for some smooth function $\psi:\Sigma\rightarrow\R$, the perturbed
surface $\eta : \Sigma \rightarrow \R^3$, $\eta(x) = S_r(x) + \nu(x)\psi(x)$,
where $\nu$ is a smooth choice of outward-pointing normal vector.
Note that we assume a-priori that the perturbed map $\eta$ is an isometric
embedding, which in turn restricts the function $\psi$.

Let us impose further that the perturbed surface $\eta$ is again critical for
the Helfrich functional with the given parameters.
We ask ourselves the question:

~\\
\hspace*{5mm}\mbox{{\bf Question.} Under which conditions will the perturbed surface $\eta$ be a sphere?}
\\

All such perturbations are termed \emph{mild}.
Note that all mild perturbations of standard round spheres have $\psi(x) =
\psi_0$ where $\psi_0\in(-r,\infty)$.
Proposition \ref{PRstab1} informs us that, in this case, any perturbation which
leaves $\eta$ at least mean convex is \emph{mild}.

Section 5 contains some initial stability analysis for spherocytes in the Helfrich model.
We summarise these results as follows.

\begin{thm}
\label{TMspheresstability}
Let $S_r:\Sigma\rightarrow\R^3$ be the standard embedding of a sphere with radius $r$.
Consider a perturbed surface $\eta:\Sigma\rightarrow\R^3$, $\eta(x) = S_r(x) +
\nu(x)\psi(x)$, with $\psi:\Sigma\rightarrow\R$ a smooth function.
Assume that $\eta$ is critical for the Helfrich functional $\SH^{c_0}$.

The perturbation $\psi$ is mild (that is, $\eta$ is a sphere) in the following cases:
\begin{enumerate}[(i)]
\item  ($c_0 = 0$) Any perturbation $\psi$ such that $\eta$ is weakly mean
convex and $\lambda$, $p$ are such that the average of the mean curvature over
the perturbed surface $\eta$ is equal to $-p/\lambda$;
\item  ($c_0 \ge 0$, $\lambda \le -k_c\frac{c_0^2}{2}$, $p\le0$) Any perturbation
$\psi$ such that $\eta$ is weakly convex;
\item  ($c_0 \ge 0$, $\lambda \le -k_c\frac{c_0^2}{2}$, $p\le -k_cc_0a_0^2$) Any
perturbation $\psi$ such that $\eta$ is weakly mean convex and on the perturbed
surface the inequality $|A^o|^2(x) \le a_0^2$, for some $a_0\in(0,\infty)$, holds.
\item  ($c_0 \ge 0$, 
$p \le -k_c\big(c_0a_0^2
               + \frac{1}{2c_0}\big(\frac{c_0^2}{2} + \frac{\lambda}{k_c}\big)^2
          \big)$)
Any perturbation $\psi$ such that on the perturbed surface the inequality
$|A^o|^2(x) \le a_0^2$, for some $a_0\in(0,\infty)$, holds.
\end{enumerate}
\end{thm}

\begin{rmk}[Mild perturbations and \cite{zhong1987instability}]
The above notion of stability, via classification of mild perturbations, is not
a-priori closely related to the classical notion of positivity of the second
variation of energy.
However, earlier studies can be cast in some cases into this framework, as we briefly now explain.

In \cite{zhong1987instability} it is shown that given a sufficiently large
pressure $p$, spherocytes may be deformed into one of a family of surfaces
associated with $l$-th order spherical harmonics.

In particular, these critical shapes are not standard round spheres.
However, some of them do appear to be convex (and also mean convex), see Figure 6 in \cite{DH76}.
Further, they have for sufficiently large pressure lower Helfrich energy than a sphere.
This gives us an interesting non-existence result for mild perturbations of spheres.

In the language we use here, this reads:

\begin{thm*}[\cite{zhong1987instability}]
Let $S_r:\Sigma\rightarrow\R^3$ be the standard embedding of a sphere with radius $r$.
Assume that the pressure $p$ is sufficiently large.
There exist smooth perturbations $\psi:\Sigma\rightarrow\R$ with the following properties:

\begin{itemize}
\item The perturbed surface $\eta:\Sigma\rightarrow\R^3$, $\eta(x) = S_r(x) +
\nu(x)\psi(x)$ is critical for the Helfrich functional $\SH^{c_0}$;
\item The perturbed surface $\eta$ is not congruent to any round sphere, and so
the perturbation $\psi$ is {\bf not mild};
\item The perturbed surface is convex.
\end{itemize}
\end{thm*}

In particular, this shows that Theorem \ref{TMspheresstability} could not hold with large pressure.
It is interesting to note that only alternative (i) allows for $p$ to be positive,
so long as $\lambda$ is negative and the average curvature condition is
satisfied.
\end{rmk}

We further note that 
\cite{pleiner} studied perturbations of the shapes identified by
\cite{zhong1987instability} for $l=2$, performing a nonlinear stability
analysis.
Bifurcations were found at vaious levels of pressure.
Bifurcations were also studied by \cite{peterson1985instability}.
\cite{safran} examined the role played by $\overline{k}$ in the stability of spheres.

Theorem \ref{TMspheresstability} highlights the crucial role that convexity
plays in the analysis of stability.
We observe that weak convexity (that each of the principal curvatures are
non-negative at every point) is a much stronger condition than weak mean
convexity (that the sum of the principal curvatures is non-negative at every
point).
Since the value of $a_0$ may be quite large depending upon the osmotic pressure
difference $p$, it may be that in practice perturbations of spherocytes fall
into or near to categories (iii) and (iv) of Theorem \ref{TMspheresstability}.

Due to this focus on convexity and mean convexity, we suspect that many of the
solutions to (P2) are mean convex, irrespective of any stability concerns, that
is, without using any knowledge about the solution being a-priori `close' to a
sphere in some sense.
This is partially confirmed in the last result of our paper:

\begin{thm}
\label{TMsolnsmeanconvex}
Suppose that $c_0 > 0$, $|A^o|^2 \le a_0^2$, $\lambda \ge k_c(a_0^2 - \frac{c_0^2}{2})$, and $p < -c_0k_ca_0^2$.
Consider a smoothly embedded closed weakly mean convex orientable surface
$f:\Sigma\rightarrow\R^3$.
Then if $f$ is critical for the Helfrich functional $\SH^{c_0}$, it must be strictly mean convex.
\end{thm}

The paper is organised as follows.
Section 2 contains some brief mathematical background required for the calculations and proofs in the later sections.
Section 3 contains the proof of Theorem \ref{TMspheresintro} as well as some discussion on minimisers.
Section 4 is concerned with rigidity, and contains the proof of Theorem \ref{mainthmx} and Theorem \ref{mainthmxx}.
Section 5 is concerned with stability analysis and the proof of Theorem \ref{TMspheresstability} and Theorem \ref{TMsolnsmeanconvex}.

\section{Mathematical background}
\label{SECTmathback}
Let us briefly set notation and describe the mathematical setting in which we work.
We are interested in properties of a human red blood cell realised as an infintely thin two-dimensional shell.
Mathematically we represent that as the image of a map $f:\Sigma\rightarrow\R^3$ with the following properties:
\begin{itemize}
\item[$\bullet$] $\Sigma$ is a smooth closed orientable differentiable manifold of dimension two; 
\item[$\bullet$] $f$ is a smooth map with injective first derivative;
\item[$\bullet$] $f$ is a homeomorphism onto its image.
\end{itemize}
An example of a smooth closed differentiable manifold of dimension two is the sphere $\S^2$.
Since we are primarily interested in the analysis of possibly spherical red blood cells, this is the primary example to
keep in mind.
For the second dot point, this is enough to imply that $f$ is an \emph{immersion}, which implies that the tensor
$g_{ij}$ with components ($\partial$ here denotes the standard partial derivative)
\[
g_{ij} = \IP{\partial_if}{\partial_jf}
\]
is a Riemannian metric.
It is the induced or pullback metric, and sometimes written as $g = f^*\delta$ where $\delta$ is the Euclidean metric
(the identity matrix).
This means that the pair $(\Sigma,g)$ is a Riemannian manifold, and $f$ is then called an \emph{isometric immersion}.
If bullet point three holds, then $f(\Sigma)$ does not have any
self-intersections, the map $f$ is injective or one-to-one, and $f$ is
upgraded to an \emph{isometric embedding}.
The red blood cell as we see it under the microscope is not $(\Sigma,g)$, but the image $f(\Sigma)$.

All geometric data can be derived from the map $f$.
At each point $p$ there is a tangent space $T_pM$ and a normal space $N_pM$.
Since the codimension of $f$ is one, the normal space is always a line.
We choose a canonical global normal vector field $\nu$ pointing out from the interior of $f(\Sigma)$.

The curvature is encoded in the second fundamental form, with components $A_{ij}$ given by
\[
A_{ij} = \IP{\partial_if}{\partial_j\nu}\,.
\]
The Weingarten equation tells us that
\[
\partial_i\nu = A_{il}g^{lm}\partial_mf = A_i^m\partial_mf\,.
\]
Coordinate independent curvature quantities that arise in the paper include the mean curvature:
\[
H = g^{ij}A_{ij} = \kappa_1 + \kappa_2\,,
\]
(in the above repeated indices are summed over, $g^{ij} = (g^{-1})_{ij}$, and $\kappa_1$, $\kappa_2$ are the eigenvalues
of $A^i_j = g^{ik}A_{kj}$), the Gauss curvature
\[
K = \text{det }A^i_j = \kappa_1\kappa_2\,,
\]
the square of the second fundamental form:
\[
|A|^2 = \kappa_1^2 + \kappa_2^2\,,
\]
and the square of the tracefree second fundamental form:
\[
|A^o|^2 = \frac12(\kappa_1-\kappa_2)^2\,.
\]
The mean curvature, the unit normal, and the Laplace-Beltrami operator are further related by
\[
\Delta f = -H\nu = \vec{H}\,.
\]
The rightmost expression is called the \emph{mean curvature vector}.

In Section \ref{SECTsphere} we use some facts about the standard sphere of radius $r>0$ embedded in $\R^3$.
Let us denote by $S_r$ an embedding $S_r:\S^2\rightarrow\R^3$ that yields a standard round sphere of radius $r>0$
centred at any point $c\in\R^3$.
The image is given by
\[
S_r(\S^2) = \{x\in\R^3\,:\,|x-c|=r\}\,,
\]
where $|y| = |(y_1,y_2,y_3)| = \sqrt{y_1^2 + y_2^2 + y_3^2}$ denotes the standard length of vectors in $\R^3$.
An exterior unit normal vectorfield to $S_r$ is given by ($x = f(p)$)
\[
\nu(x) = \frac{x-c}{|x-c|} = \frac{x-c}{r} = \frac{f(p)-c}{r}\,.
\]
From the Weingarten equation we find
\[
H = g^{ij}A_{ij}
 = g^i_jA_i^j
 = g^{ij}A_i^mg_{mj}
 = g^{ij}\IP{\partial_i\nu}{\partial_jf}
 = \text{div }\nu\,,
\]
and so for the sphere of radius $r$,
\[
H = g^{ij}\IP{\partial_if}{\partial_jf}\frac{1}{r}
  = g^{ij}g_{ij}\frac{1}{r}
  = \frac{1}{r}\text{Trace }I_2
  = \frac{2}{r}\,.
\]
We note that this calculation is invariant under translation: the centre point
$c$ of the sphere does not play any role in the curvature.
Since the sphere with centre translated back to the origin is
$SO(3)$-invariant, each of the principal curvatures $\kappa_1$ and $\kappa_2$
are equal and so the above implies $\kappa_1 = \kappa_2 = \frac1r$.
In particular, we have
\[
K = \frac1{r^2}\quad\text{and}\quad |A^o|^2 = 0\,.
\]
The last condition is necesssary and sufficient: if $f$ is a sphere, then
$|A^o|^2 = 0$, and if $|A^o|^2 = 0$ and $f$ is closed, then $f$ is a sphere
(see Proposition 8.2.9 in \cite{pressley}).
This leads one naturally to consider the range of non-zero values that $|A^o|^2$ may
take to signify a kind of `distance' from being spherical.

In Section \ref{SECTsphere} we also use some elementary facts about the calculus of variations and the existence of
minimisers.
In particular, there the following fact is used.

\begin{lem}
Suppose that the functional $\hat{\SH}^{c_0}$ (or $\SH^{c_0}$) is unbounded from below for a given choice of parameters
$c_0$, $S_0$, $V_0$ (or $c_0$, $\lambda$, $p$) in the class of smooth embeddings.
Then there does not exist a solution to (P1) (or (P2)).
\label{LMfact}
\end{lem}
The proof is straightforward and standard.
\begin{proof}
By hypothesis, there exists a sequence $\{f_i\}_{i\in\N}$ of smooth embedded orientable surfaces such that the energy
$E(f_i)\rightarrow-\infty$ as $i\rightarrow\infty$.
Here we have used $E$ to denote either $\hat{\SH^{c_0}}$ or $\SH^{c_0}$.
If a solution $f_0$ to either problem were to exist, then the minimisation condition \eqref{EQmin1} (or \eqref{EQmin2})
would be satisfied.
However since $f_0$ is smooth, it has finite energy, and so
\[
E(f_0) > E(f_j)
\]
for some $j$ sufficiently large.
This is a contradiction.
\hfill $\Box$
\end{proof}
\begin{rmk}
The hypothesis that the functional be unbounded from below only has to hold for one particular sequence.
\end{rmk}

In Section \ref{SECTlocal} we need the following elementary result from differential geometry.

\begin{lem}
\label{LMmclem}
Suppose $f:\Sigma\rightarrow\R^3$ is a smooth embedded orientable surface containing the origin.
Then there exists at least one point where the mean curvature of $f$ is strictly positive.
\end{lem}
\begin{proof}
Consider the function $|f|^2$.
Since $\Sigma$ is closed, $|f|^2$ achieves a global maximum on $\Sigma$.
At this point the Hessian of $|f|^2$ is non-positive.
We compute
\[
\text{Hess}_{ij}|f|^2 = \nabla_i\nabla_j |f|^2 = 2g_{ij} + \IP{f}{\nabla_i\nabla_jf}\,.
\]
Tracing the above with $g$, at a maximum we have
\begin{equation}
\label{LMmclemEQ1}
0 \ge \Delta|f|^2 = 4 - \IP{f}{\nu}H\,.
\end{equation}
Since $\nu$ is an outward-pointing unit normal and the origin is contained in the interior of $f(\Sigma)$, at a global
maximum of $|f|^2$ we have $\IP{f}{\nu} = c > 0$.
Therefore we conclude from \eqref{LMmclemEQ1} that
\[
H \ge \frac4{\IP{f}{\nu}}
\]
at a global maximum of $|f|^2$.
\hfill $\Box$
\end{proof}


\section{Spherical solutions}
\label{SECTsphere}
Let $S_{r}:\S^2\rightarrow\R^3$ be the embedding of the standard sphere with
radius $r$ and centre at the origin.
Recalling the basic properties of spheres explained in Section
\ref{SECTmathback}, the Euler-Lagrange equation
\eqref{EQeulerlagrange} evaluated at $S_r$ is quadratic in $1/r$:
\begin{equation}
\label{EQprequadratic}
2k_cc_0r^{-2} - \Big(\frac{k_cc_0^2}{2} + \lambda\Big)2r^{-1} - p = 0\,.
\end{equation}
If $c_0 \ne 0$ (we deal with the case $c_0 = 0$ in Case 0 below) then
\eqref{EQprequadratic} is equivalent to
\begin{equation}
\label{EQquadratic}
r^{-2} - \Big(\frac{c_0}{2} + \frac{\lambda}{k_cc_0}\Big)r^{-1} - \frac{p}{2k_cc_0} = 0\,.
\end{equation}
We set $r_{\pm}$ to be the roots, if they exist, of this quadratic; that is
\begin{equation}
\label{EQradii}
r_{\pm}
 =
  \frac12\bigg(
   \frac{\lambda}{k_cc_0} + \frac{c_0}{2} \pm \sqrt{\Big(\frac{c_0}{2}
 + \frac{\lambda}{k_cc_0}\Big)^2 + \frac{2p}{k_cc_0}}
       \ \bigg)\,.
\end{equation}
Given that $S_r$ exists only for $r>0$, there is at least one spherical
critical point of the Helfrich energy if
\begin{equation}
\label{EQradiusofsphere}
 \frac{\lambda}{k_cc_0} + \frac{c_0}{2}
 + \sqrt{\Big(\frac{c_0}{2} + \frac{\lambda}{k_cc_0}\Big)^2 + \frac{2p}{k_cc_0}} > 0\,.
\end{equation}
Let us set
\[
	u = \frac{\lambda}{k_cc_0} + \frac{c_0}{2}\,,\quad\text{and}\quad
	v = \frac{2p}{k_cc_0}\,.
\]
For any spherical solution to exist, the argument of the square root in Equation \eqref{EQradiusofsphere} must be non-negative, that is,
\begin{equation}
	\label{EQatleastone}
v \ge -u^2
\,.
\end{equation}
We separate now into three cases:
\begin{align*}
\text{Case 0}&: c_0 = 0\\
\text{Case 1}&: u = -u^2\\
\text{Case 2}&: u > -u^2\,.
\end{align*}

\subsection{Case 0: $c_0 = 0$}
In this case the problem greatly simplifies: Formula \eqref{EQprequadratic} reads
\[
2\lambda r^{-1} + p = 0\,.
\]
If $\lambda = p = 0$ then we are dealing with the degenerate case of the
Willmore functional and any sphere $S_r$ of any radius is a critical point (in
fact a minimiser, see \cite{Willmore65}).

If $\lambda = 0$ and $p\ne0$ then there does not exist a spherical solution.
Otherwise, we find
\[
r = -\frac{2\lambda}{p}\,.
\]
For the right hand side to be positive, we require $\lambda$ and $p$ to be of
opposite signs.
If this is the case the sphere $S_{-2\lambda/p}$ is the (only) spherical solution.

We summarise this in the following lemma.

\begin{lem}[Resolution of Case 0]
\label{LMc0zero}
Suppose $f:\Sigma\rightarrow\R^3$ is a closed, smooth, embedded orientable surface in the
same topological class as a sphere.
Suppose $c_0 = 0$.
If $f(\Sigma) = S_r(\Sigma)$ and $f$ is critical for the Helfrich functional $\SH^{c_0}$, then
one of the following must hold:
\begin{enumerate}
\item[(i)]
$\lambda = p = 0$, in which case any $r\in(0,\infty)$ is possible; or
\item[(ii)]
$p\lambda < 0$, in which case $r = -\frac{2\lambda}{p}$.
\end{enumerate}
\end{lem}

\begin{rmk}[Minimisers I]
If we work under the additional assumption that $f$ is in fact a minimiser,
then we may refine the parameters above and in later lemmata throughout this
section. In particular, the functional reads
\begin{align*}
\SH^{c_0}(S_r)
 &=  \frac{k_c}{2} \int_\Sigma (H-c_0)^2d\mu 
   + \lambda\text{Area }\Sigma
   + p\text{Vol }\Sigma
   + 2\overline{k}\pi\chi(\Sigma)
\\
 &=  4k_c\pi
   - 8k_cc_0\pi r
   + 4\bigg(\lambda+\frac{k_cc_0^2}{2}\bigg)\pi r^2
   + \frac{4p\pi}{3} r^3
   + 4\overline{k}\pi
   \,.
\end{align*}
If $p<0$, then $\SH^{c_0}(S_r) \rightarrow -\infty$ as $r\rightarrow\infty$.
In light of Lemma \ref{LMfact} this implies that no minimiser exists, spherical or otherwise.
Adding the hypothesis that $f$ solves (P2) and thus is a minimiser would remove
the possibility that $p<0$ in alternative (ii), leaving only $\lambda < 0$.
It removes the possibility that $p<0$ in later lemmata as well.

One must note that this argument succeeds in producing anything simply due to
the lack of scaling in our interpretation of the Helfrich model.  As noted
earlier, $c_0$, $\lambda$ and $p$ should be functions of the configuration
$f:\Sigma\rightarrow\R^3$, and as such should depend on scale.  Therefore one
should not take this remark and the others (Minimisers II -- IV) to be anything
more than purely mathematical observations.
\end{rmk}

\subsection{Case 1: $v = -u^2$}
In this case the problem again greatly simplifies: The square root in \eqref{EQradii} drops out and so $r_+ = r_-$.
Further, inequality \eqref{EQradiusofsphere} simplifies to $u > 0$.
Then by \eqref{EQradii} we have $r = \frac{u}{2}$.

The lemma for this case is as follows.
\begin{lem}[Resolution of Case 1]
Suppose $f:\Sigma\rightarrow\R^3$ is a closed, smooth, embedded orientable surface in the
same topological class as a sphere.
If $f(\Sigma) = S_r(\Sigma)$, $v = -u^2$, and $f$ is critical for the Helfrich functional $\SH^{c_0}$, then $u > 0$ and
\[
	r = \frac{u}{2} = \sqrt{\frac{-p}{2k_cc_0}}\,.
\]
\label{LMcase1}
\end{lem}

\begin{rmk}[Minimisers II]
	If $f$ is a minimiser then by the previous remark $p\ge0$.
	Then since $v < 0$ in this case we must have $c_0 < 0$.
	Also, as $u > 0$ this implies $\lambda < 0$.
	That is, the only remaining possibility that a sphere is minimising
	here requires $p \ge0$, $c_0 < 0$ and $\lambda < 0$.
\end{rmk}

\subsection{Case 2: $v > -u^2$}
In this case inequality \eqref{EQatleastone} must be satisfied for us to have at least one spherical solution.
That is,
\[
	u + \sqrt{u^2 + v} > 0\,.
\]
We further separate into the two subcases where $u \ge 0$ and $u < 0$.

\subsubsection{Case 2.1: $u \ge 0$}

In this case we automatically have $r_+ > 0$ and so one critical sphere exists. For a second to exist we require by \eqref{EQradii}
\begin{align*}
	u - \sqrt{u^2 + v} &> 0\text{, or}
	\\
	u^2 &> u^2 + v\text{, since $u > 0$}
\end{align*}
which implies that $v < 0$, or $v \in (-u^2,0)$.

\subsubsection{Case 2.2: $u < 0$}

In this case we have at most one critical sphere. Since $u = -|u|$, for this single sphere to exist we require
\begin{align*}
	u + \sqrt{u^2 + v} &> 0\text{, or}
	\\
	-|u| &> -\sqrt{u^2 + v}\text{, which implies}
	\\
	v &> 0\,.
\end{align*}

We summarise the results of Case 2 in the following lemma.
\begin{lem}
\label{LMcase2}
Suppose $f:\Sigma\rightarrow\R^3$ is a closed, smooth, embedded orientable surface in the
same topological class as a sphere.
If $f(\Sigma) = S_r(\Sigma)$, $v > -u^2$, and $f$ is critical for the Helfrich functional $\SH^{c_0}$, then one of the following must hold:
\begin{enumerate}
	\item[(i)]   $u \ge 0$ and $v\ge0$, in which case $f$ must be the unique critical sphere with radius $r = \frac12(u + \sqrt{u^2 + v})$; or
	\item[(ii)]  $u \ge 0$ and $v\in(-u^2,0)$, in which $f$ may be either of the two critical spheres with radii $r_\pm = \frac12(u \pm \sqrt{u^2 + v})$; or
	\item[(iii)] $u < 0$ and $v>0$, in which case $f$ must be the unique critical sphere with radius $r = \frac12(u + \sqrt{u^2 + v})$.
\end{enumerate}
\end{lem}

\begin{rmk}[Minimisers III]
	Supposing that $f$ is a minimiser rules out $p<0$ as before.
	In (ii) it is dramatic: $v < 0$ and so we must have $c_0 < 0$. As
	$u\ge0$, this implies $\lambda < 0$.
	For (i) it is not as useful since $v \ge 0$ implies $c_0 \ge 0$, but
	this doesn't combine with $u\ge0$ to yield a sign condition on
	$\lambda$; it only implies $\lambda \ge -\frac{k_cc_0^2}{2}$.
	For (iii) we do obtain a sign restriction, since $v > 0$ implies $c_0 >
	0$, and so $u<0$ implies $\lambda < 0$.
\end{rmk}

~\\
{\it Proof of Theorem \ref{TMspheresintro}.}
Combine Lemmata \ref{LMc0zero} -- \ref{LMcase2}.
\qed~\\

Although the theorem details many possible circumstances under which spherical critical surfaces exist, there is a
three-dimensional family of parameters and so one should interpret this as roughly stating that spherical biomembranes
may only occur in very special situations.

We expect that several of the spherocytes identified by Theorem
\ref{TMspheresintro} have high energy and are unstable.
For a more rigorous interpretation of this intuition see the remarks after Lemmata \ref{LMc0zero} -- \ref{LMcase2}.

Therefore we do not expect to observe these spherocytes often in live experiments.
Some work in identifying stability properties is presented in Sections 4 and 5.
In these sections, we see that deformations which preserve a notion of convexity,
or do not alter the magnitude of the difference of the principal curvatures in
too great a manner, are mild.
The work in Section 4 shows that there exist deformations of
spherocytes that do not preserve convexity that are also mild.
(This is because a non-convex embedded sphere may nevertheless satisfy
$\vn{A^o}_2^2 < \varepsilon$.)

One exception to this generic instability appears to be the case when $c_0 > 0$
and $r = \frac2{c_0}$, so that $H = c_0$ and the curvature integral in
$\SH^{c_0}$ on
spheres with radius $r$ vanishes.
This case was earlier identified in \cite{Peterson89} as having energy independent of $c_0$.
It is in fact quite special and spheres are fundamental with those parameters: spheres are global minimisers of the
energy and thus have the best stability possible.

\begin{thm}
Suppose $c_0 > 0$ and $\lambda = p = 0$.
Then the unique global minimal solution to (P2) in the class of smooth embedded orientable surfaces is a sphere of radius
$\frac{2}{c_0}$.
\end{thm}
\begin{proof}
For these choices of $c_0$, $\lambda$ and $p$ the functional reads
\begin{equation}
\label{EQfunctsph}
\SH^{c_0}(f)
 =  \frac{k_c}{2} \int_\Sigma (H-c_0)^2d\mu 
   + 2\overline{k}\pi\chi(\Sigma)
 =  \frac{k_c}{2} \int_\Sigma (H-c_0)^2d\mu 
   + 4\overline{k}\pi\,.
\end{equation}
The Euler-Lagrange equation is
\begin{equation}
\label{EQeulerlagrangesph}
k_c(\Delta H + H|A^o|^2) + 2k_cc_0K - \frac{k_cc_0^2}{2}H = 0\,.
\end{equation}
The sphere with $r = \frac2{c_0}$ has $H = c_0$ and $K = \frac{c_0^2}{4}$.
Further $|A^o| = 0$ and $\Delta H = 0$ on any sphere, so $f = S_r : \Sigma \rightarrow \R^3$ solves
\eqref{EQeulerlagrangesph}.

This proves that $S_r$ is a critical point for the functional in \eqref{EQfunctsph}.
Minimality is easy to see, since the functional is the sum of a non-negative integral and a constant.
The sphere $S_r$ has $H=c_0$ and so the integral in \eqref{EQfunctsph} takes on its lowest possible value: zero.

Let us also prove that up to translation and rotation it is unique.
If $f:\Sigma\rightarrow\R^3$ is any smooth embedded orientable surface, then
\[
 \frac{k_c}{2} \int_\Sigma (H-c_0)^2d\mu 
 = 0
\quad\Longrightarrow\quad H = c_0\,.
\]
A classical theorem of \cite{AlexCMCintheLarge} tells us that any embedded orientable surface with constant mean
curvature must be a sphere, with radius $r = \frac{2}{H}$, which is exactly what we wanted.
\hfill $\Box$
\end{proof}

In the next section below we extend this reasoning using additionally the
isoperimetric inequality to allow some cases where $\lambda$ and $p$ do not
vanish.

\section{Rigidity}

\subsection{Proof of Theorem \ref{mainthmx}}

Let $f:\Sigma\rightarrow\R^3$ be a closed smooth surface in the same topological class as a sphere.
We combine two famous results on the uniqueness of the standard round sphere.
The first is the Hopf theorem.
This states that:
\begin{equation}
\label{hopf}
\text{If $f$ is an embedded closed CMC surface, then it must be a standard round sphere.}
\end{equation}
The second is the isoperimetric theorem.
It states that:
\begin{equation*}
   (\text{Area }\Sigma)^3 \ge 36\pi(\text{Vol }\Sigma)^{2}
\end{equation*}
or
\begin{equation}
\label{isop}
   \text{Area }\Sigma \ge (36\pi)^\frac13(\text{Vol }\Sigma)^{\frac23}
\end{equation}
with equality if and only if $f$ is a standard round sphere.

Fix $\V = \frac{32\pi}{3c_0^3}$.
For $f \in \mathcal{F}$, using the above and condition \eqref{mainthmxcondn}, we estimate the energy by
\begin{align*}
(\SH^{c_0}(f) - 4\overline{k}\pi) &= \frac{k_c}{2}\int_\Sigma (H-c_0)^2\,d\mu + \lambda \A + p\V
\\
 &\ge \Big[\lambda (36\pi)^\frac13 + p\Big(\frac{4\pi}{3c_0^3}\Big)^\frac13\Big]\V^\frac23
\\
 &\ge \Big[3\lambda  + \frac{p}{c_0}\Big]\Big(\frac{4\pi}{3}\Big)^\frac13\V^\frac23
\\
 &\ge 0
\end{align*}
with equality only when $f$ is a standard round sphere with the given
prescribed volume, that is, any sphere with radius $\frac2{c_0}$.

\subsection{Proof of Theorem \ref{mainthmxx}}

{\bf Note.} In this section the constant $c\in\R$ may change from line to
line, referring to an absolute constant with each usage.

We begin by writing Corollary 5 in \cite{MW13} in our notation for closed
surfaces and with the choice $\gamma \equiv 1$:

\begin{lem}
Let $f:\Sigma\rightarrow\R^3$ be a closed, smooth, embedded orientable surface.
Then
\begin{align*}
&\int_\Sigma \big(|\nabla_{(2)}A|^2
                + |A|^2|\nabla A|^2
                + |A|^4|A^o|^2\big)\,d\mu
   + \Big(c_0^2 + 2\frac{\lambda}{k_c}\Big)\int_\Sigma |\nabla H|^2\,d\mu
\\
&\le 
    c\int_\Sigma \BH^{c_0}(f)\cdot\Delta H\,d\mu
  + c_1\int_\Sigma\big(|A^o|^6 + |A^o|^2|\nabla A^o|^2\big)\,d\mu
\\&\qquad
  + c\,c_0\int_\Sigma (\Delta H)|A^o|^2\,d\mu
  - c\,\frac{c_0}{2}\int_\Sigma (\Delta H)H^2\,d\mu
\,,
\end{align*}
where $c$ and $c_1$ are absolute constants.
\label{GapLem}
\end{lem}

Now we work to estimate the additional terms involving the spontaneous curvature on the right hand side.

\begin{lem}
Let $f:\Sigma\rightarrow\R^3$ be a closed, smooth, embedded orientable surface.
Then
\begin{align*}
c\,c_0 \int_\Sigma (\Delta H)|A^o|^2\,d\mu
&\le \frac12\int_\Sigma |\nabla_{(2)}A|^2\,d\mu
\\&\quad
 + c_2\,c_0^2\cdot(\A)\vn{A^o}_2^2\int_\Sigma |\nabla_{(2)}A|^2 + |A|^4|A^o|^2\,d\mu\,.
\end{align*}
where $c$ and $c_2$ are absolute constants.
\label{GapLem2}
\end{lem}
\begin{proof}
Since $H = g^{ij}A_{ij}$ and $\Delta = g^{ij}\nabla_i\nabla_j$, we can use the
Cauchy-Schwarz inequality to see that
\begin{equation}
\label{EQstart}
|\Delta H| \le |\nabla_{(2)}A|
\,,
\end{equation}
where on the left we have the absolute value of the scalar function $\Delta H$
and on the right we have the induced norm of the $(0,4)$-tensor field with components
$\nabla_i\nabla_j A_{kl}$.

Using the inequality $2ab \le a^2 + b^2$ we find
\begin{equation}
\label{EQss}
  c\,c_0\int_\Sigma (\Delta H)|A^o|^2\,d\mu
\le
 \frac12\int_\Sigma |\nabla_{(2)}A|^2\,d\mu
 + \frac12c\,c_0^2\int_\Sigma |A^o|^4\,d\mu\,.
\end{equation}
Applying the Michael-Simon Sobolev inequality (see \cite{MSS}) we find
\[
\int_\Sigma |A^o|^4\,d\mu
 \le c\bigg(
       \int_\Sigma |\nabla A^o|\,|A^o| + |H|\,|A^o|^2\,d\mu
      \bigg)^2\,.
\]
To each term on the right hand side we use the H\"older inequality, yielding
\begin{equation}
\label{EQmid}
\int_\Sigma |A^o|^4\,d\mu
 \le c\vn{A^o}_2^2\int_\Sigma |\nabla A^o|^2\,d\mu
   + c\vn{A^o}_2^2\int_\Sigma |H|^2\,|A^o|^2\,d\mu
\,.
\end{equation}
For the first term, we estimate
\begin{equation}
\label{EQmidd}
 c\vn{A^o}_2^2\int_\Sigma |\nabla A^o|^2\,d\mu
\le \frac{1}{4\A} \vn{A^o}_2^4 + c\cdot(\A)\vn{\nabla A^o}_2^4\,.
\end{equation}
The divergence theorem and closedness implies
\[
\vn{\nabla A^o}_2^2 = -\IP{A^o}{\Delta A^o}_{L^2} \le \vn{A^o}_2 \vn{\Delta A^o}\,.
\]
Since $A^o = A - \frac12 gH$, estimate \eqref{EQstart} implies $|\Delta A^o|
\le 2|\nabla_{(2)}A|$.
H\"older's inequality implies that $\vn{A^o}_2^4 \le (\A)\vn{A^o}_4^4$.
Combining these with \eqref{EQmidd}, we find
\begin{equation}
\label{EQmiddd}
 c\vn{A^o}_2^2\int_\Sigma |\nabla A^o|^2\,d\mu
\le \frac{1}{4} \vn{A^o}_4^4 + c\cdot(\A)\vn{A^o}_2^2\vn{\nabla_{(2)}A}_2^2\,.
\end{equation}
Now for the second term on the right hand side of \eqref{EQmid}, we estimate similarly
\begin{align}
c\vn{A^o}_2^2\int_\Sigma |H|^2\,|A^o|^2\,d\mu
 &\le \frac{1}{4} \vn{A^o}_4^4 + c\cdot(\A)\bigg(\int_\Sigma |H|^2\,|A^o|^2\,d\mu\bigg)^2
\notag
\\
\label{EQmide}
 &\le \frac{1}{4} \vn{A^o}_4^4 + c\cdot(\A)\vn{A^o}_2^2\int_\Sigma |A|^4\,|A^o|^2\,d\mu
\,.
\end{align}
Now we combine \eqref{EQmiddd}, \eqref{EQmide} with \eqref{EQmid} and absorb to conclude 
\begin{equation}
\label{EQlate}
\int_\Sigma |A^o|^4\,d\mu
 \le c\cdot(\A)\vn{A^o}_2^2\vn{\nabla_{(2)}A}_2^2
   + c\cdot(\A)\vn{A^o}_2^2\int_\Sigma |A|^4\,|A^o|^2\,d\mu
\,.
\end{equation}
Together with the estimate \eqref{EQss} this finishes the proof.
\end{proof}

\begin{prop}
Let $f:\Sigma\rightarrow\R^3$ be a closed, smooth, embedded orientable surface.
There exist universal constants $c_1, c_2, c_3, c_4$ such that if
\begin{equation}
\label{EQsmall}
\int_\Sigma |A^o|^2\,d\mu \le \min\bigg\{
                                    \frac1{2c_1c_3},
                                    \frac1{2c_0^2c_2(\A)}
                                  \bigg\}
                          := \varepsilon_2
\end{equation}
then
\begin{align*}
&\int_\Sigma \big(|\nabla_{(2)}A|^2
                + |A|^2|\nabla A|^2
                + |A|^4|A^o|^2\big)\,d\mu
   + \Big(c_0^2 + 2\frac{\lambda}{k_c}\Big)\int_\Sigma |\nabla H|^2\,d\mu
\\
&\le 
    c\int_\Sigma \BH^{c_0}(f)\cdot\Delta H\,d\mu
  + c_4\,c_0^2(\A)\int_\Sigma |\nabla_{(2)}A|^2 + |A|^2\,|\nabla A|^2\,d\mu
\,,
\end{align*}
where $c$ is an absolute constant.
\label{GapLemProp}
\end{prop}
\begin{proof}
We note the following consequence of the Michael-Simon Sobolev inequality,
first proven in Lemma 2.5 of \cite{KS02}:
\begin{align}
\int_\Sigma &\big(|\nabla A^o|^2|A^o|^2 + |A^o|^6)\gamma^4 d\mu
\notag\\
&\le
  c_3\vn{A^o}^2_{2}\int_\Sigma |\nabla_{(2)}A|^2 + |\nabla A|^2|A|^2 + |A|^4|A^o|^2 \, d\mu
\,.
\label{EQmssforAo}
\end{align}
By assuming that $\varepsilon_2 \le \frac1{2c_1c_3}$ we may absorb this term from the right
hand side of Lemma \ref{GapLem} into the left and multiply through to obtain
\begin{align*}
&\int_\Sigma \big(|\nabla_{(2)}A|^2
                + |A|^2|\nabla A|^2
                + |A|^4|A^o|^2\big)\,d\mu
   + \Big(c_0^2 + 2\frac{\lambda}{k_c}\Big)\int_\Sigma |\nabla H|^2\,d\mu
\\
&\le 
    c\int_\Sigma \BH^{c_0}(f)\cdot\Delta H\,d\mu
  + c\,c_0\int_\Sigma (\Delta H)|A^o|^2\,d\mu
  - c\,\frac{c_0}{2}\int_\Sigma (\Delta H)H^2\,d\mu
\,.
\end{align*}
Second, by assuming that $\varepsilon_2 \le
\min\{(2c_1c_3)^{-1},(2c_0^2c_2(\A))^{-1}\}$ and absorbing with the help of
Lemma \ref{GapLem2}, we find
\begin{align*}
&\int_\Sigma \big(|\nabla_{(2)}A|^2
                + |A|^2|\nabla A|^2
                + |A|^4|A^o|^2\big)\,d\mu
   + \Big(c_0^2 + 2\frac{\lambda}{k_c}\Big)\int_\Sigma |\nabla H|^2\,d\mu
\\
&\le 
    c\int_\Sigma \BH^{c_0}(f)\cdot\Delta H\,d\mu
  - c\,\frac{c_0}{2}\int_\Sigma (\Delta H)H^2\,d\mu
\,.
\end{align*}
For the last remaining term we integrate by parts and estimate to obtain
\begin{equation}
\label{EQlatee}
  - c\,\frac{c_0}{2}\int_\Sigma (\Delta H)H^2\,d\mu
= 
    c\,c_0\int_\Sigma H|\nabla H|^2\,d\mu
\le \frac12\int_\Sigma |\nabla A|^2|A|^2\,d\mu
 + c\,c_0^2 \int_\Sigma |\nabla H|^2\,d\mu
\,.
\end{equation}
Now the identity (6) from \cite{MW13} implies
\[
 c\,c_0^2 \int_\Sigma |\nabla H|^2\,d\mu
\le 4c\,c_0^2\int_\Sigma |\nabla^* A^0|^2\,d\mu
\le  4c\,c_0^2\int_\Sigma |\nabla A^0|^2\,d\mu
\,.
\]
Using the Michael-Simon Sobolev inequality we find
\begin{align}
4c\,c_0^2\int_\Sigma |\nabla A^0|^2\,d\mu
&\le c\,c_0^2\bigg(\int_\Sigma |\nabla_{(2)}A| + |H|\,|\nabla A|\,d\mu\bigg)^2
\notag
\\
&\le c\,c_0^2(\A)\int_\Sigma |\nabla_{(2)}A|^2 + |A|^2\,|\nabla A|^2\,d\mu
\,.
\label{EQlateee}
\end{align}

Absorbing with \eqref{EQlatee}, and bounding the extra term as in
\eqref{EQlateee}, our main estimate reads
\begin{align}
&\int_\Sigma \big(|\nabla_{(2)}A|^2
                + |A|^2|\nabla A|^2
                + |A|^4|A^o|^2\big)\,d\mu
   + \Big(c_0^2 + 2\frac{\lambda}{k_c}\Big)\int_\Sigma |\nabla H|^2\,d\mu
\notag
\\
&\le 
    c\int_\Sigma \BH^{c_0}(f)\cdot\Delta H\,d\mu
  + c_4\,c_0^2(\A)\int_\Sigma |\nabla_{(2)}A|^2 + |A|^2\,|\nabla A|^2\,d\mu
\,,
\label{EQrefmainest}
\end{align}
where $c_4$ is an absolute constant.
The estimate \eqref{EQrefmainest} is the statement of the Proposition.
\end{proof}

In Proposition \ref{GapLemProp} we note the quantity $c_0^2(\A)$; this is
dimensionless, as $c_0$ has the units of the mean curvature, and $\A$ has the
units of $d\mu$. Therefore $c_0^2(\A)$ has the units of $H^2d\mu$, and is
scale-invariant.

\begin{cor}
Let $f:\Sigma\rightarrow\R^3$ be a closed, smooth, embedded orientable surface.
Assume that $f$ is a Helfrich surface, that is,
\begin{equation*}
\BH^{c_0}(f) = 0
\,.
\end{equation*}
There exist universal constants $c_1, c_2, c_3, c_4$ such that if
\eqref{EQsmall} holds, $c_0^2 + 2\frac{\lambda}{k_c} \ge 0$, and
\begin{equation}
\label{EQc0small}
c_0^2(\A) < \frac1{c_4}
\end{equation}
then $f(\Sigma) = S_r(\Sigma)$. 
\end{cor}
\begin{proof}
The smallness condition \eqref{EQcrit} allows us to absorb the second term in
\eqref{EQrefmainest} into the left hand side. This yields
\begin{align*}
&\int_\Sigma \big(|\nabla_{(2)}A|^2
                + |A|^2|\nabla A|^2
                + |A|^4|A^o|^2\big)\,d\mu
   + \Big(c_0^2 + 2\frac{\lambda}{k_c}\Big)\int_\Sigma |\nabla H|^2\,d\mu
\\
&\le  
    c\int_\Sigma \BH^{c_0}(f)\cdot\Delta H\,d\mu
\,.
\end{align*}
The condition $c_0^2 + 2\frac{\lambda}{k_c} \ge 0$, and \eqref{EQcrit} upgrade this to
\begin{equation*}
\int_\Sigma \big(|\nabla_{(2)}A|^2
                + |A|^2|\nabla A|^2
                + |A|^4|A^o|^2\big)\,d\mu
\le  
    0
\,.
\end{equation*}
With smoothness of $f$, this implies $|A^o| = 0$, and we are done.
\end{proof}

\section{Local stability}
\label{SECTlocal}

We begin with a proof of Proposition \ref{PRstab1} from the introduction.

~\\
{\it Proof of Proposition \ref{PRstab1}.} 
In this case we have
\[
\Delta H + H|A^o|^2 = 0\,.
\]
Integrating and using the divergence theorem we have
\[
\int_\Sigma \Delta H + H|A^o|^2 d\mu = \int_\Sigma H|A^o|^2 d\mu = 0\,.
\]
Since the surface is smooth, the mean curvature $H$ and the norm squared of the
trace-free second fundamental form $|A^o|^2$ are smooth.  Due to the mean
convexity hypothesis the only possibility is that $|A^o|^2(x) = 0$ wherever
$H(x) > 0$.
By Lemma \ref{LMmclem} there is at least one $x_0$ where $H(x_0) > 0$.
Therefore $|A^o|^2(x_0) = 0$.
Now recall the definition of $A^o$:
\[
A^o_{ij} = A_{ij} - \frac12g_{ij}H\,.
\]
The Codazzi equation tells us that the tensor $\nabla_iA_{jk}$ is totally symmetric.
Therefore, using the summation convention and the definition $H = g^{ij}A_{ij}$
\begin{equation}
\label{Codazzi}
\nabla^k A^o_{kj} = \nabla^kA_{kj} - \frac12\nabla_jH = \frac12\nabla_jH\,.
\end{equation}
By smoothness, $H(x) > 0$ in a neighbourhood of $x_0$.
Denote the maximal such neighbourhood by $\Omega$.
Suppose that $\partial\Omega \ne \emptyset$.
Note that $H(x) = 0$ for all $x\in\partial\Omega$.
Now the above argument shows that $|A^o|^2(x) = 0$ on $\Omega$, and so by
\eqref{Codazzi} above we have 
\[
\nabla_jH = \nabla^k A^o_{kj} = 0
\quad\text{ on }\quad\Omega
\]
for any $j$. Therefore $H(x) = H(x_0)$ is constant on $\Omega$.
However, $H(x) = 0$ on $\partial\Omega$, and therefore $H$ must be discontinuous on $\partial\Omega$.
This is a contradiction.
Therefore $\partial\Omega=\emptyset$, $\Omega = \Sigma$, and $H$ is constant on all of $\Sigma$.
Since $f:\Sigma\rightarrow\R^3$ is an embedding, the theorem of \cite{AlexCMCintheLarge} applies and we conclude that $f$ is a sphere.
\qed
~\\

Proposition \ref{PRstab1} yields a stability statement in the following sense.
Let $S_r:\Sigma\rightarrow\R^3$ be a standard sphere with radius $r$.
Then it solves the problem (P2) with $c_0 = \lambda = p = 0$.
Consider, for some smooth function $\psi:\Sigma\rightarrow\R$, the perturbed
surface $\eta : \Sigma \rightarrow \R^3$, $\eta(x) = S_r(x) + \nu(x)\psi(x)$,
where $\nu$ is a smooth choice of normal vector.
Let us impose that the perturbed surface $\eta$ also solves the problem (P2)
with the given parameters.
We ask ourselves the question: under which conditions on the perturbation
$\psi$ will $\eta$ be a sphere?
All such perturbations are termed \emph{mild}.
Proposition \ref{PRstab1} informs us that, in this case, any perturbation which
leaves $\eta$ at least mean convex is \emph{mild}.

As witnessed in Section \ref{SECTsphere}, variations on the parameters
$\lambda$, $p$ and $c_0$ induce wild changes in the behaviour of the solutions
to (P2).
In general, we do not expect spherical solutions to be stable.
For certain ranges of these parameters, we are nevertheless able to obtain a
result analogous to that of Proposition \ref{PRstab1}.

\begin{prop}
\label{PRstab2}
Consider a smoothly embedded closed weakly mean convex surface $f:\Sigma\rightarrow\R^3$.
Suppose that $c_0 = 0$ and that $p, \lambda$ are such that the average of the mean curvature of $f$ is equal to $-p/\lambda$.
Then if $f$ is critical for the Helfrich functional $\SH^{c_0}$, it must be a sphere.
\end{prop}
\begin{proof}
In this case we have
\[
\Delta H + \Big(|A^o|^2 - \frac{\lambda}{k_c}\Big)H - \frac{p}{k_c} = 0\,.
\]
Rearranging, this implies
\[
\Delta H + H|A^o|^2 = \frac{\lambda}{k_c}H + \frac{p}{k_c}
\,.
\]
Observe that the integral of the right hand side vanishes:
\[
\int_\Sigma \bigg(\frac{\lambda}{k_c}H + \frac{p}{k_c}\bigg)d\mu
 = k_c^{-1} \bigg[ \lambda \int_\Sigma H\,d\mu + p|\Sigma| \bigg]
 = k_c^{-1} \bigg[ \lambda |\Sigma| \frac{-p}{\lambda} + p|\Sigma| \bigg]
 = 0\,.
\]
Therefore the proof of Proposition \ref{PRstab1} goes through analogously in this case.
\end{proof}

It is possible to extend this integral method in various directions to obtain results specific for narrow choices of the parameters $\lambda$, $p$ and $c_0$.
What we wish to do now is to illustrate a different method that appears more suitable to the case where $c_0 \ne 0$.
It has the drawback of requiring either weak convexity (as opposed to weak {\bf mean} convexity above) or a condition on $|A^o|^2$.
The method relies on the following standard tool.
The statement below is a corollary of the more general theorem proved in \cite{CalabiHopf58} (see also \cite{Hopf27}).

\begin{thm}[Calabi-Hopf Maximum Principle]
\label{TMcalabi}
Suppose $(\Sigma, g)$ is a Riemannian manifold.
Consider $u:U\rightarrow\R$ a smooth function defined over the open set $U\subset\Sigma$.
If
\[
(\Delta u)(x) \le 0
\]
everywhere in $U$, and if $u$ attains a local minimum value at some point in $U$, then $u$ is identically constant in $U$.
\end{thm}

The theorem allows the following pair of corollaries.

\begin{cor}
\label{CY1}
Consider a smoothly embedded closed weakly convex surface $f:\Sigma\rightarrow\R^3$.
Suppose that $c_0 \ge 0$, $p \le 0$, and $\lambda \le -k_c\frac{c_0^2}{2}$.
Then if $f$ is critical for the Helfrich functional $\SH^{c_0}$, it must be a sphere.
\end{cor}
\begin{proof}
In this case we have
\begin{align*}
\Delta H &= - \Big(|A^o|^2 - \frac{c_0^2}{2} - \frac{\lambda}{k_c}\Big)H - 2c_0K + \frac{p}{k_c}
\\
       &\le \Big(\frac{c_0^2}{2} + \frac{\lambda}{k_c}\Big)H - 2c_0K + \frac{p}{k_c}\,.
\end{align*}
The weak convexity hypothesis means that at every point the principal curvatures $\kappa_1, \kappa_2$ are non-negative.
This implies that $-2c_0K := -2c_0\kappa_1\kappa_2 \le 0$.
The mean curvature is also clearly non-negative, and the conditions on $c_0$,
$\lambda$, $p$ imply that the entire right hand side is non-positive.

Therefore we have
\[
(\Delta H)(x) \le 0
\]
at every point $x\in\Sigma$.
Note that $\Sigma$ is an open set inside $\Sigma$, and since it is additionally
compact and $H$ is a smooth function on $\Sigma$, it must achieve its minimum
at some point in $\Sigma$.
Theorem \ref{TMcalabi} applies, yielding $H$ identically constant, and again
\cite{AlexCMCintheLarge} shows that $f$ must be a standard embedding of the
sphere $S_r:\Sigma\rightarrow\R^3$.

\end{proof}

\begin{cor}
Consider a smoothly embedded closed weakly mean convex surface $f:\Sigma\rightarrow\R^3$.
Let $a_0 > 0$ be such that
\[
|A^o|^2(x) \le a_0^2
\]
for every $x\in\Sigma$.
Suppose that $c_0 \ge 0$, $p \le -k_cc_0a_0^2$, and $\lambda \le -k_c\frac{c_0^2}{2}$.
Then if $f$ is critical for the Helfrich functional $\SH^{c_0}$, it must be a sphere.
\label{CY2}
\end{cor}
\begin{proof}
Similarly to the proof above, we have
\begin{equation}
\label{EQcy2}
\Delta H \le \Big(\frac{c_0^2}{2} + \frac{\lambda}{k_c}\Big)H - 2c_0K + \frac{p}{k_c}\,.
\end{equation}
The general strategy has not changed -- our goal remains to show that under the given hypotheses, the right hand side of the above differential inequality is non-positive.
Recall that, in terms of the principal curvatures $\kappa_1$, $\kappa_2$, we have:
\[
|A|^2 = \kappa_1^2 + \kappa_2^2;\qquad
H^2 = (\kappa_1 + \kappa_2)^2;\qquad
|A^o|^2 = \frac12(\kappa_1-\kappa_2)^2;\qquad
K = \kappa_1\kappa_2\,,
\]
so that $K = \frac14H^2 - \frac12|A^o|^2$.
Therefore
\[
-2c_0K + \frac{p}{k_c}
 = -\frac12c_0H^2 + \Big(\frac{p}{k_c} + c_0|A^o|^2\Big)
 \le 0\,,
\]
by hypothesis.

The weak mean convexity hypothesis means that $H \ge 0$, and so we have again
that the conditions on $c_0$, $\lambda$, $p$ imply that the entire right hand
side of \eqref{EQcy2} is non-positive.
The proof now continues analogously to that of Corollary \ref{CY1}.
\end{proof}

If $p$ is sufficiently negative, then one may use this maximum principle idea
to remove all geometric assumptions apart from embeddedness.
We present a final variation of the idea above incorporating this observation.

\begin{thm}
Consider a smoothly embedded closed surface $f:\Sigma\rightarrow\R^3$.
Let $a_0 > 0$ be such that
\[
|A^o|^2(x) \le a_0^2
\]
for every $x\in\Sigma$.
Suppose that $c_0 \ge 0$, and 
\[
p \le -k_c\bigg(c_0a_0^2
               + \frac{1}{2c_0}\Big(\frac{c_0^2}{2} + \frac{\lambda}{k_c}\Big)^2
          \bigg)
\,.
\]
Then if $f$ is critical for the Helfrich functional $\SH^{c_0}$, it must be a sphere.
\end{thm}
\begin{proof}
Using the calculations in the proof above, we factorise
\begin{equation*}
\Delta H \le -\bigg(
                \frac{\sqrt{c_0}}{\sqrt{2}}H + \frac{1}{\sqrt{2c_0}}\Big(\frac{c_0^2}{2} + \frac{\lambda}{k_c}\Big)
              \bigg)^2
               + \frac{1}{2c_0}\Big(\frac{c_0^2}{2} + \frac{\lambda}{k_c}\Big)^2
               + c_0a_0^2 + \frac{p}{k_c}\,.
\end{equation*}
The hypothesis implies that the right hand side is non-positive.
The proof now continues analogously to that of Corollary \ref{CY1}.
\end{proof}

Combining these four results gives alternatives $(i)$ through to $(iv)$ of Theorem \ref{TMspheresstability}.

\begin{rmk}[Minimisers IV]
	Proposition \ref{PRstab1} is for Willmore surfaces, and alternatives (ii) -- (iv)
	require $p < 0$.  This means that although a certain class of
	perturbations are mild, there must be other perturbations that decrease
	the energy.

	Therefore from the perspective of minimisers, the
	most interesting alternative is (i), in the particular case where $p >
	0$ and $\lambda < 0$.
\end{rmk}

Clearly one may tweak the geometric and parametric conditions that allow the
two strategies outlined in the above stability results to go through.
It appears that some form of convexity is critical to the argument, and so it
is interesting to determine conditions on the parameters under which convexity
automatically holds.
We finish this section with a demonstration of how one may deduce such a
result.
\\

{\it Proof of Theorem \ref{TMspheresstability}.} 
Rearrranging the Euler-Lagrange equation \eqref{EQeulerlagrange} we find
\[
\Delta H + \Big(|A^o|^2 - \frac{c_0^2}{2} - \frac{\lambda}{k_c}\Big)H + 2c_0K - \frac{p}{k_c} = 0
\]
Suppose a global minimum for $H$ occurs at $x$.
Setting $H = H(x)$, we have
\[
\Big(|A^o|^2 - \frac{c_0^2}{2} - \frac{\lambda}{k_c}\Big)H
 \le
    \frac{p}{k_c}
    - 2c_0K
\,.
\]
Using $2K = \frac12H^2 - |A^o|^2$ as in the proof of Corollary \ref{CY2}, at $x$ we have
\[
\Big(|A^o|^2 - \frac{c_0^2}{2} - \frac{\lambda}{k_c}\Big)H
 \le
    \frac{p}{k_c}
    - \frac{c_0}{2}H^2 + c_0|A^o|^2
\,.
\]
Setting $a = \frac{c_0}{2}$, $b = |A^o|^2 - \frac{c_0^2}{2} - \frac{\lambda}{k_c}$, and $c = -\frac{p}{k_c} - c_0|A^o|^2$ the above is
\begin{equation}
\label{EQparabola}
P(H) = aH^2 + bH + c \le 0\,.
\end{equation}
In $H$, $P$ is a parabola with zero, one, or two real roots.

{\bf Zero roots.} In this case either $P$ is always positive or always negative.
Since $c_0 > 0$, the term $aH^2$ dominates for large enough $H$ and $P$ will be
positive there.
This contradicts \eqref{EQparabola}, which holds for $H = H(x)$.

{\bf One root.} In this case $P$ is always non-positive or always non-negative.
As with the case above, $c_0 > 0$ implies that for large enough $H$, $P$ is
positive.
Therefore $H$ is the unique point where $P$ touches the axis, that is $H$ solves
\[
P'(H) = 0 \qquad \Longleftrightarrow \qquad H = -\frac{b}{2a}.
\]
By hypothesis $|A^o|^2 \le a_0^2$, and $\lambda \ge k_c(a_0^2 - \frac{c_0^2}{2})$, so  
\begin{equation}
\label{EQin1}
b = |A^o|^2 - \frac{c_0^2}{2} - \frac{\lambda}{k_c}
 \le a_0^2 - \frac{c_0^2}{2} - \frac{k_c(a_0^2 - \frac{c_0^2}{2})}{k_c} = 0\,,
\end{equation}
and thus $H$ is non-negative.
Further, $H = 0$ only if $b = 0$, in which case $-4ac = 0$, since $b^2 - 4ac =
0$ when we have only one root.
As $a = \frac{c_0}{2} > 0$, this implies that $c = 0$.
However our hypothesis $p < -c_0k_ca_0^2$ implies
\begin{equation}
\label{EQin2}
c \ge -\frac{p}{k_c} - c_0a_0^2 > 0\,,
\end{equation}
a contradiction.
Therefore $H$ is strictly positive.

{\bf Two roots.} In this case $P$ changes sign. There are two roots $H_1$ and
$H_2$ given by the quadratic formula.
We may assume that $H_1 < H_2$.
Since $P(H) \le 0$ and $a > 0$, $H$ will lie in the interval $[H_1,H_2]$.
The lower bound $H_1$ satisfies
\[
H_1 = \frac{-b-\sqrt{b^2-4ac}}{2a}\,.
\]
The computation \eqref{EQin1} shows that $b$ is non-positive, and \eqref{EQin2} shows that $c$ is strictly positive.
Therefore
\begin{align}
\notag
\qquad ac &> 0
\\
\notag
\Longrightarrow\qquad b^2 &> b^2 - 4ac
\\
\notag
\Longrightarrow\qquad -b &> \sqrt{b^2 - 4ac}\,.
\end{align}
Therefore $H_1 > 0$.

We conclude that the mean curvature in each case is strictly positive, and so
the minimum of the mean curvature in $\Sigma$ is strictly positive, as
required.
\qed


\section*{Acknowledgements}
The third author would like to thank Carsten Hartmann for introducing her to the Helfrich model and biomembranes.
The second author would like to thank Annette Worthy for helpful discussions related to this work.
The second and third authors were supported by ARC grant DP120100097.

The first author was partially supported by URC Small Grant 228381024 during
two visits to the University of Wollongong, where this work was partially
completed.
The first author would also like to thank the FIM at ETH Z\"urich and SFB DFG
71 for financial support during the completion of this work.


\bibliographystyle{spmpscinat}
\bibliography{MathBiol}

\end{document}